%% file: main.tex
\documentclass{article}
\usepackage{arxiv}

\title{Faster Privacy-Preserving Computation of Edit Distance with Moves}

\author{
  Yohei Yoshimoto\\
  Kyushu Institute of Technology\\
  680-4 Kawazu, Iizuka, Fukuoka 820-8502, Japan \\
  \texttt{y\_yoshimoto@donald.ai.kyutech.ac.jp} \\
   \And
  Masaharu Kataoka\\
  Kyushu Institute of Technology\\
  680-4 Kawazu, Iizuka, Fukuoka 820-8502, Japan \\
  \texttt{m\_kataoka@donald.ai.kyutech.ac.jp} \\
   \And
  Yoshimasa Takabatake\\
  Kyushu Institute of Technology\\
  680-4 Kawazu, Iizuka, Fukuoka 820-8502, Japan \\
  \texttt{takabatake@donald.ai.kyutech.ac.jp} \\
   \And
  Tomohiro I\\
  Kyushu Institute of Technology\\
  680-4 Kawazu, Iizuka, Fukuoka 820-8502, Japan \\
  \texttt{tomohiro@donald.ai.kyutech.ac.jp} \\
   \And
  Kilho Shin\\
  Gakushuin University\\
  1-5-1 Mejiro, Toshimaku, Tokyo 171-8588, Japan \\
  \texttt{kilhoshin314@gmail.com} \\
   \And
  Hiroshi Sakamoto\\
  Kyushu Institute of Technology\\
  680-4 Kawazu, Iizuka, Fukuoka 820-8502, Japan \\
  \texttt{hiroshi@donald.ai.kyutech.ac.jp} \\
}
\usepackage[utf8]{inputenc} % allow utf-8 input
\usepackage[T1]{fontenc}    % use 8-bit T1 fonts
\usepackage{hyperref}       % hyperlinks
\usepackage{booktabs}       % professional-quality tables
\usepackage{amsfonts}       % blackboard math symbols
\usepackage{nicefrac}       % compact symbols for 1/2, etc.
\usepackage{microtype}      % microtypography
\usepackage{lipsum}		% Can be removed after putting your text content

\usepackage{amsmath,amssymb}
\usepackage[pdftex]{graphicx}
\usepackage{wrapfig}
\usepackage{array}
\usepackage[algoruled,linesnumbered,algo2e,vlined]{algorithm2e}
\usepackage{amssymb}
\usepackage{url}
\usepackage{multirow}
\usepackage{cite}

\usepackage{algorithm}
\usepackage{algorithmic}

\usepackage{comment}

 \newtheorem{theorem}{\bf Theorem}
 \newtheorem{definition}{Definition}
 \newtheorem{proof}{Proof}
\begin{document}
\maketitle
\begin{abstract}
We consider an efficient two-party protocol for securely computing the similarity of strings
w.r.t. an extended edit distance measure.
Here, two parties possessing strings $x$ and $y$, respectively,
want to jointly compute an approximate value for $\mathrm{EDM}(x,y)$, 
the minimum number of edit operations including substring moves needed to transform $x$ into $y$,
without revealing any private information.
Recently, the first secure two-party protocol for this was proposed, based on homomorphic encryption,
but this approach is not suitable for long strings due to its high communication and round complexities.
In this paper, we propose an improved algorithm
that significantly reduces the round complexity without sacrificing its cryptographic strength.
We examine the performance of our algorithm for DNA sequences compared to previous one.
\end{abstract}

\input{sec1}

\input{sec2}

\input{sec3}

\input{sec4}

\input{sec5}

\bigskip
\noindent
\textbf{Acknowledgments.}
This work was supported by JST CREST (JPMJCR1402),
KAKENHI (16K16009, 17H01791, 17H00762 and 18K18111) and Fujitsu Laboratories Ltd. 
The authors thank anonymous reviewers for their helpful comments.

\bibliographystyle{plain}
\bibliography{mybib}

\end{document}

%% file: sec1.tex
\section{Introduction}

\subsection{Motivation}
As the number of strings containing personal information has increased, 
privacy-preserving computation has become more and more important.
Secure computation based on public key encryption is one of the great achievements of modern cryptography,
as it enables untrusted parties to compute a function based on 
their private inputs while revealing nothing but the result.

In addition, edit distance is a well-established metric for measuring the similarity or dissimilarity of two strings.
The rapid progress of gene sequencing technology has expanded the range of edit distance applications
to include personalized genomic medicine, disease diagnosis, and preventive treatment (for example, see~\cite{Akgun2015}).
A person's genome is, however, ultimately individual information that 
uniquely identifies its owner, so the parties involved should not share their personal genomic data as plaintext.

Thus, we consider a secure multi-party edit distance computation based on the public key encryption model.
Here, untrusted two parties generating their own public and private keys 
have strings $x$ and $y$, respectively, and want to jointly compute $f(x,y)$
for a given metric $f$ without revealing anything about their individual strings.

\subsection{Related Work}
Homomorphic encryption (HE) based on the public key encryption model is an emerging technique 
that is being used for secure multi-party computation.
The Paillier encryption system~\cite{Paillier1999} possesses additive homomorphism,
enabling us to perform additive operations on two encrypted integers {\em without decryption}.
This means that parties can jointly compute the encrypted value $E(x+y)$ directly based only on
two encrypted integers $E(x)$ and $E(y)$.
%\footnote{In this example, if one party can obtain the decrypted sum $z=x+y$, they immediately know the other integer. Preventing this issue will of course require additional techniques.}.
HE is also {\em probabilistic}, i.e.,
an adversary can hardly predicts $x$ given $E(x)$, even if they can observe some number of $(x', E(x'))$ pairs for any $x'$.

By taking advantage of these characteristics, researchers have proposed several HE-based 
privacy-preserving protocols for computing the Levenshtein distance $d(x,y)$.
For example, Inan et al.~\cite{Inan2007} designed a three-party protocol where
two parties securely compute $d(x,y)$ by enlisting the help of a reliable third party.
Rane and Sun~\cite{Rane2010} then improved this three-party protocol to develop the first two-party one.

In this paper, we focus on an interesting metric called the {\em edit distance with moves} (EDM),
where we allow any substring to be moved with unit cost in addition to the standard Levenshtein distance operations.
Based on the EDM, we can find a set of approximately maximal common substrings appearing in two strings,
which can be used to detect plagiarism in documents or long repeated segments in DNA sequences.
As an example, consider two unambiguously similar strings $x=a^Nb^N$ and $y=b^Na^N$,
which can be transformed into each other by a single move.
Whereas the exact EDM is simply $\mathrm{EDM}(x,y)=1$,
the Levenshtein distance has the undesirable value $d(x,y)=2N$.
The $n$-gram distance is preferable to the Levenshtein's in this case,
but it requires huge time/space complexity depending on $N$.

Although computing $\mathrm{EDM} (x,y)$ is NP-hard~\cite{Shapira07},
Cormode and Muthukrishnan~\cite{Cormode2007} were able to find an almost linear-time approximation algorithm for it.
Many techniques have been proposed for computing the EDM;
for example, Ganczorz et al.~\cite{Ganczorz2018} proposed a lightweight probabilistic algorithm.
In these algorithms, each string $x$ is transformed into a characteristic vector $v_x$ consisting of nonnegative integers
representing the frequencies of particular substrings of $x$.
For two strings $x$ and $y$, we then have the approximate distance guaranteeing 
$L_1(v_x,v_y) = O(\lg^*N\lg N )\mathrm{EDM}(x,y)$ 
for $N=|x|+|y|$.\footnote{In Appendix~A of~\cite{HSP}, the authors point out that there is a subtle flaw in the ESP algorithm~\cite{Cormode2007} that achieves this $O(\lg^*N\lg N )$ bound. However, this flaw can be remedied by an alternative algorithm called HSP~\cite{HSP}.}
Since $\lg^*N$ increases extremely slowly\footnote{$\lg^*N$ is the number of times the logarithm function $\lg$ must be iteratively applied to $N$ until the result is at most 1.},
we employ $L_1(v_x,v_y)$ as a reasonable approximation to $\mathrm{EDM}(x,y)$.

Recently, Nakagawa et al. proposed 
the first secure two-party protocol for EDM (sEDM)~\cite{Nakagawa2018} based on HE,
but, their algorithm suffers from a bottleneck during the step where the parties construct a shared labeling scheme.
This motivated us to improve the previous algorithm to make it easier to use in practice.

\subsection{Our Contribution}

\begin{table*}[t]
\begin{center}
\caption{
Comparison of the communication and round complexities of secure EDM computation models.
Here, $N$ is the total length of both parties' input strings,
$n$ is the number of characteristic substrings determining the approximate EDM, and
$m$ is the range of the rolling hash $H(\cdot)$ for the substrings satisfying $m>n$.
``Naive'' is the baseline method that uses $H(\cdot)$ as the labeling function for the characteristic substrings.
In this table, we omit the security parameter or the unit cost of encryption and decryption
because the models use a same key length (e.g., 256-bit).
} 
%\vspace{1mm}
\label{tab1}
\begin{tabular}{l|ll}
\hline
\qquad Method \qquad\qquad & \qquad Communication \qquad\qquad & Round \qquad\qquad \\ \hline
\qquad Ours & \qquad $O(n\lg n + m)$ & $O(1)$ \\
\qquad Naive & \qquad $O(m\lg m)$ & $O(1)$ \\ 
\qquad sEDM~\cite{Nakagawa2018} & \qquad $O(n\lg n)$ & $O(\lg N)$ \\ \hline
\end{tabular}
\end{center}
%\vspace{-3mm}
\end{table*}

The complexities of our algorithm and related ones are summarized in Table~\ref{tab1}.
Computing the approximate EDM involves two phases: the shared labeling of characteristic substrings (Phase 1) and 
the $L_1$-distance computation of characteristic vectors (Phase 2).
First, we outline those phases below.
Let the parties have strings $x$ and $y$, respectively. 
In the offline case (i.e., there is no need for privacy-preserving communication), 
they construct the respective parsing trees $T_x$ and $T_y$ by the bottom-up parsing called ESP~\cite{Cormode2007} where
the node labels must be {\em consistent}, i.e., two labels are equal if they correspond to the same substring.
In such an ESP tree, a substring derived by an internal node is called a characteristic substring.
In a privacy-preserving model, the two parties should jointly compute such consistent labels without revealing whether or not 
a characteristic substring is common to both of them (Phase 1). 
After computing all the labels in $T_x$ and $T_y$, they jointly compute the $L_1$-distance
of two characteristic vectors consisting the frequencies of all labels in $T_x$ and $T_y$ (Phase 2).

As reported in~\cite{Nakagawa2018}, in terms of usefulness, a bottleneck exists in Phase 1.
The task is to design a bijection $f: X\cup Y \to \{1,2,\ldots, n\}$
where $X$ and $Y$ ($|X\cup Y|=n$) are the sets of characteristic substrings
for the parties, respectively.
Since $X$ and $Y$ are computable without communication,
the goal is to jointly compute $f(w)$ for any $w\in X$ without revealing whether or not $w\in Y$.
Here, this problem is closely related to the private set operation (PSO) where
parties possessing their private sets want to obtain the results for several set operations, e.g., intersection or union.
Applying the Bloom filter~\cite{Bloom1970} and HE techniques, various protocols for PSO have been proposed~\cite{Kissner2005,Blanton2012,Davidson2017}.
However, these protocols cannot be directly applicable to our problem because
these protocols require at least three parties for the security constrained.
Thus, we propose a novel secure two-party protocol for Phase 1.

As shown in Table~\ref{tab1}, 
we eliminate the $O(\lg N)$ round complexity using the proposed method that can achieve $O(1)$
round complexity while maintaining the efficiency of communication complexity.
Furthermore, we examine the practical performance of our algorithm for real DNA sequences.

%% file: sec2.tex
\section{Preliminaries}

\subsection{EDM}

\begin{figure*}[t]
\begin{center}
\includegraphics[width=0.75\textwidth]{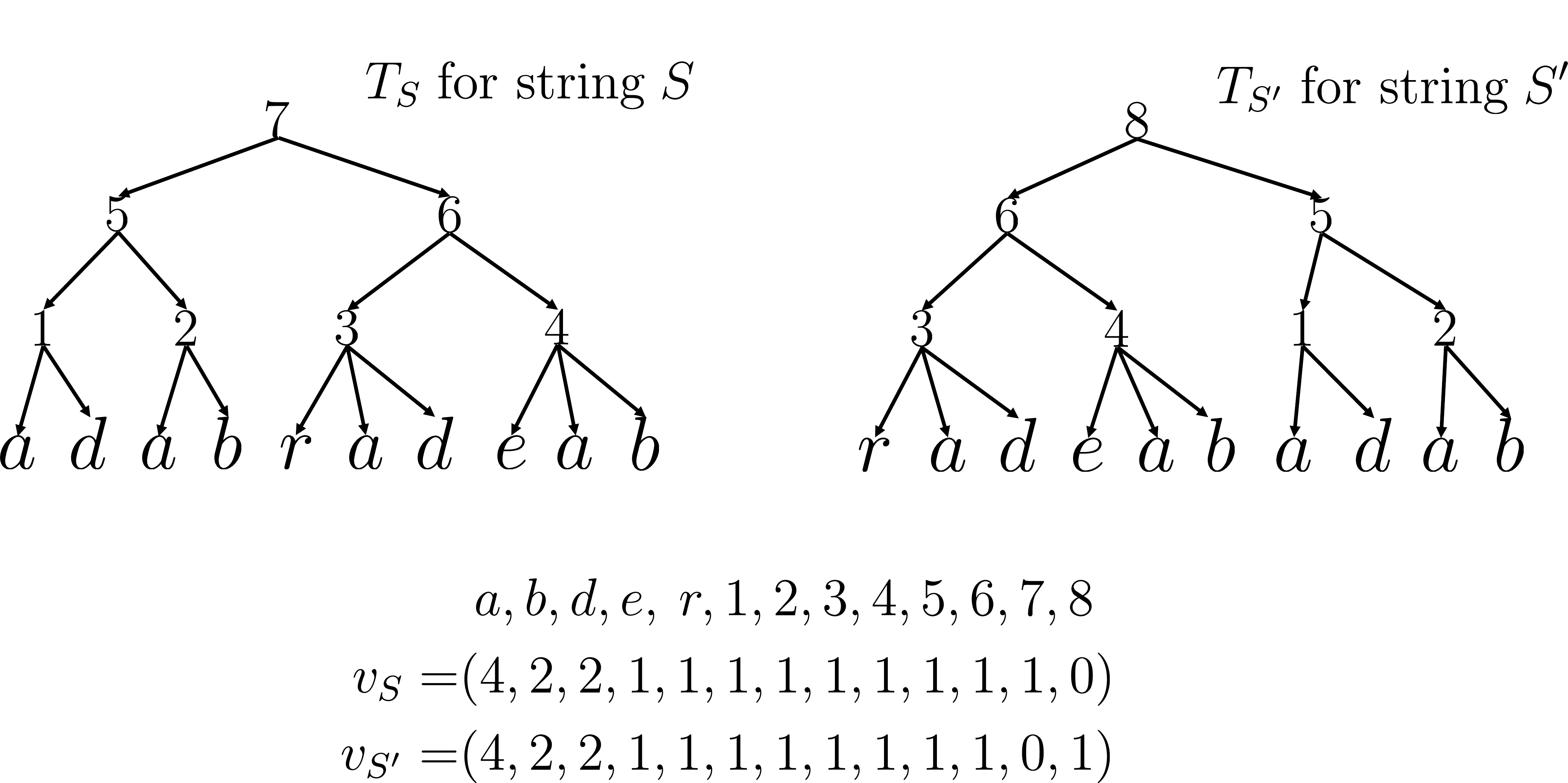}
\end{center}
%\vspace{-1.2cm}
\caption{Example of approximate EDM computation for strings $S$ and $S'$.
Here, the ESP trees $T_S$ and $T_{S'}$ are constructed by applying a shared labeling scheme for all internal nodes.
After constructing $T_S$ and $T_{S'}$, the corresponding characteristic vectors $v_S$ and $v_{S'}$ are computed offline. 
Finally, the exact $\mathrm{EDM}(S,S')$ is approximated by $L_1(v_S,v_{S'})=2$.}
\label{fig:esp}
\end{figure*}

Let $\Sigma$ be a finite set of alphabet symbols and $\Sigma^*$ be its closure.
Denote the set of all strings of the length $N$ by $\Sigma^N$ and
the length of a string $S$ by $|S|$.
For simplicity, we also denote the cardinality of a set $U$ by $|U|$.
In addition, $S[i]$ denotes the $i$-th symbol of $S$ and $S[i..j]$ denotes the substring $S[i]S[i+1]\cdots S[j]$.

Next, we define $\mathrm{EDM}(S,S')$ as the length of the shortest sequence of edit operations
that transforms $S$ into $S'$, where the permitted operations (each with unit cost) are inserting, deleting, or renaming
one symbol at any position and moving an arbitrary substring.
Unfortunately, as Theorem~\ref{ESP-1} shows,
computing $\mathrm{EDM}(S,S')$ is NP-hard even if renaming operations are not allowed~\cite{Shapira07},
so we focus on an approximation algorithm for EDM, called ESP (Edit-Sensitive Parsing)~\cite{Cormode2007}.

\begin{theorem}{\rm (Shapira and Storer \cite{Shapira07})}\label{ESP-1}
Determining $\mathrm{EDM}(x,y)$ is NP-hard 
even if only three unit-cost operations namely inserting or deleting a character and moving a substring are allowed.
\end{theorem}

To illustrate the ESP algorithm, we consider a string $S\in\Sigma^*$.
First, $S$ is deterministically partitioned into blocks as $S=s_1s_2\cdots s_k$ such that $2\leq |s_i|\leq 3$.
Here, we omit the details since partitions are determined based on $S$ alone, without communication.
Next, a {\em consistent} label $\ell(s_i)$ is assigned to each block $s_i$, where $\ell(x)=\ell(y)$ if $x=y$.
The resulting string $L=\ell(s_1)\cdots \ell(s_k)$ is then processed recursively until $|L|=1$.
Finally, a parsing tree $T_S$ is obtained for $S$,
which can be used to approximate $EDM(S,S')$ by applying the following result.

\begin{theorem}{\rm (Cormode and Muthukrishnan \cite{Cormode2007})}\label{ESP-2}
Let $T_S$ and $T_{S'}$ be consistently labeled ESP trees for $S,S'\in\Sigma^*$, and
let $v_S$ be the characteristic vector for $S$, where $v_S[k]$ is the frequency of label $k$ in $T_S$.
Then, 
\[\frac{\: 1\:}{2}\mathrm{ EDM}(S,S') \leq L_1(v_S,v_{S'}) = O(\lg^*N\lg N)\mathrm{EDM}(S,S')\] 
for 
 $L_1(v_S,v_{S'})= \displaystyle \sum_{i=1}^k|v_{S}[i]-v_{S'}[i]|$.
% where $\lg^*N$ is the number of times the logarithm $\lg$ is iteratively applied until the result is less than or equal to 1.
\end{theorem}

Figure~\ref{fig:esp} shows an example of applying consistent labeling
to the trees $T_S$ and $T_{S'}$, together with the resulting characteristic vectors.
The strings $S$ and $S'$ are partitioned offline, so the problem of 
preserving privacy reduces to designing secure protocol for creating consistent labels and
computing the $L_1$-distance between the trees.
In this study, we propose a novel HE-based algorithm for the former problem.

\subsection{Homomorphic Encryption}

\newcommand\G{\mathbb G}
\newcommand\Zp[1]{\mathbb Z_{#1}}

Here, we briefly review the framework of homomorphic encryption.
Let $(pk,sk)$ be a key pair for  a public-key encryption scheme, and let
$E_{pk}(x)$ be the encrypted values of message $x$
and $D_{sk}(C)$ be the decrypted value of ciphertext $C$, respectively.
We say that the encryption scheme is {\em additively homomorphic} if we have the properties:
(1) There is an operation $h_+(\cdot,\cdot)$ for $E_{pk}(x)$ and $E_{pk}(y)$ such that
$D_{sk}(h_+(E_{pk}(x),E_{pk}(y)))=x+y$.
(2) For any $r$, we can compute the scalar multiplication such that $D_{sk}(r\cdot E_{pk}(x))=r\cdot x$.

An additive homomorphic encryption scheme allowing sufficient number of these operations is called an additive HE\footnote{In general, the number of applicable operations over ciphertexts is bounded by the size of $(pk,sk)$.}.
Paillier's encryption scheme~\cite{Paillier1999} is the first secure additive HE,
but we cannot evaluate many functions by only the additive homomorphism and scalar multiplication.

On the other hand, the multiplication 
$D_{sk}(h_\times(E_{pk}(x), E_{pk}(y)))=x\cdot y$ is another important homomorphism.
If we allow both additive and multiplicative homomorphism as well as scalar multiplication (called a
fully homomorphic encryption, FHE~\cite{Gentry2009} for short),
it follows that we can perform any arithmetic operation on ciphertexts.
For example, if we can use sufficiently number of additive operations and a single multiplicative operation over ciphertexts,
we obtain the inner-product of two encrypted vectors.

However, there is a tradeoff between the available homomorphic operations and their computational cost.
To avoid this difficulty, we focus on the Leveled HE (LHE) where the number of homomorphic multiplications is restricted beforehand.
In particular, {\em two-level} HE (Additive HE that allows a single homomorphic multiplication) has attracted a great deal of attention.
BGN encryption system is the first two-level HE invented by Boneh et al.~\cite{Boneh2005}
assuming a single multiplication and sufficient numbers of additions.
Using the BGN, we can securely evaluate formulas in disjunctive normal form (DNF).
After this pioneering study, many practical two-level HE protocols have been proposed~\cite{Freeman2010,Herold2014,Catalano2015,Attrapadung2018}.

For the EDM computation, 
Nakagawa et al.~\cite{Nakagawa2018} introduced an algorithm for computing the EDM based on two-level HE,
but their algorithm is very slow for large strings.
So, we propose another novel secure computation of EDM for large strings based on the faster two-level HE
proposed by Attrapadung et al.~\cite{Attrapadung2018}.
As far as we know, there are no secure two-party protocols for the EDM computation that only use additive homomorphic property. 
Whether we can compute EDM on a two-party protocol based on additive HE only is an interesting question.

%% file: sec3.tex
\section{Two-Party Secure Consistent Labeling}
\label{sec:two-party-secure}

\subsection{Hash Function}
\label{sec:hash-function}

In our protocol, two parties, ${\cal A}$ (Alice) and ${\cal B}$ (Bob), 
agree to use a shared hash function to assign tentative labels to their ESP trees.
First, we consider the conditions that a hash function should satisfy.
One desirable property of any hash function used for our algorithm is that
its hash value should be uniformly distributed, and we assume this is true in this paper. 
In addition, the function involves a parameter $m$ that represents the number of possible hash values, 
and this affects our algorithm's computational complexity, as well as the hash function's conflict resistance and one-wayness.

\begin{description}
\item[Computational complexity.]
In our algorithm, $\mathcal A$ encrypts $m$ individual bits and sends the resulting $m$ ciphertexts to $\mathcal B$,
who then adds or multiplies pairs of ciphertexts.
Thus, our first requirement is that $m$ be small enough for these computation to be performed efficiently.
\item[Conflict resistance.]
  The hash function's conflict resistance affects the accuracy of the edit distance estimated by Algorithm~1.
  We say that a conflict occurs when
  two distinct texts happen to be hashed to the same value.
  Very roughly,  
  if conflicts occur with probability $p$, 
  the average proportional error in the edit distances is also $O(p)$.
  That said,
  we can create conditions where the probability of conflict is below some threshold $p$, as follows.
  Let $n$ denote the number of labels (hash values) computed by the algorithm.
  To avoid conflicts, 
  $n$ must be sufficiently small relative to $m$.
  After computing $n$ hash values at random, we can estimate 
  the probability of at least one conflict having occurred as follows.
  \begin{align*}
    \Pr[\text{Conflicts}]
    & = 1 - \Pr[\text{No conflict}] \\
    % & = 1 - \frac mm \cdot \frac {m-1}m \cdots \frac{m-n+1}m \\
    & = 1 - \left(1 - \frac 1m\right) \left(1 - \frac 2m\right) \cdots
        \left(1 - \frac {n-1}m\right) \\
    % & \approx 1 - \exp\left(- \frac 1m - \frac 2m - \dots - \frac{n-1}m\right) \\
    & \approx 1- \exp\left(-\frac{n^2}{2m}\right).
  \end{align*}
  Thus,
  to ensure this probability is below a given (small) threshold $p$, we require
  \begin{equation}\label{eq:1}
    n \le - \ln\left( 1 - p\right) \sqrt{2m}.
  \end{equation}
%  is required.
\item[One-wayness.]
  One-wayness is important for security.
  For example,
  if $\mathcal A$ happens to have two texts with hash values $h$ and $h+2$, respectively, 
  but does not have a text with hash value $h + 1$, 
  then $\mathcal A$ would know that $\mathcal B$ has a text with hash value $h+1$.
  If the hash function is not one-way,
  $\mathcal A$ could then guess the next.   
  A function is theoretically one-way if it is computationally difficult
  to find an $x$ such that $H(x) = y$ given $y$ with non-negligible probability. 
  Given an ideal hash oracle that selects $H(x)$ uniformly at random from $\{1, \dots, m\}$ for any $x$,
  the probability of any guess $x'$ being correct for an unknown $x$ is exactly $\frac 1m$.
  Thus,
  to ensure the function is effectively one-way,
  $m$ must be sufficiently large.  
\end{description}

It is known that, 
if the problem of finding a pair of distinct inputs
that hash to the same value is computationally intractable (strong conflict-resistance),
the hash function is also one-way.
For cryptographic hash algorithms, such as, MD5 and SHA-1, 
strong conflict-resistance is required and the conflict probability must be negligibly small, 
such as less than $\frac 1{2^{100}}$.
This indeed requires that $m$ be very large ($2^{128}$ and $2^{160}$ for MD5 and SHA-1, respectively),
so it would be computationally unfeasible to use a cryptographic hash function in our algorithm.

Nevertheless, if we relax the requirements somewhat, 
it is not difficult in practice to select an $m$ that meets our needs.
For example, if $n = 100$ and $p = 0.05$, then $m = 1,900,416$ satisfies inequality~\eqref{eq:1}.
Generating and transmitting so many ciphertexts would be time-consuming but still feasible,
and would reduce the probability of breaking one-wayness to a very low value. 

We should, however, note that
these requirements on $m$ are merely necessary conditions for conflict resistance and one-wayness.
Even under these conditions, using a well-designed hash algorithm is still crucial.  

The rolling hash algorithm~\cite{Karp87}, 
defined as follows based on two parameters $m$ and $b$,
is expected to be sufficiently conflict-resistant and one-way.
For a given input $x = (s_1, \dots, s_\ell) \in [0, b)^\ell$,
the hash function is given by $H(x) = \sum_{i=1}^\ell s_i \cdot b^{\ell - i} \bmod m$.
This algorithm has the useful advantage 
that we can compute $H(xy)$ from $H(x)$ and $H(y)$ in constant time,
independent of the lengths of $x$ and $y$.  

\subsection{Algorithm}

\begin{algorithm}   
\caption{for consistently labeling $T_{\cal A}$ and $T_{\cal B}$}                    
%\caption{for consistent labeling of $T_{\cal A}$ and $T_{\cal B}$ in minimum name space}         
\label{algo1}                          
\begin{algorithmic}  
\STATE
\STATE {\bf Preprocessing (tentative labeling):} 
Parties $\cal A$ and $\cal B$ agree to use a shared hash function $H$ with a range $\{0, \dots, m\}$,
where $m$ is chosen so as to meet the requirements given in Section~\ref{sec:hash-function}.
Both parties compute the ESP trees $T_{\cal A}$ and $T_{\cal B}$ corresponding to their respective strings offline,
then assign labels $H(w)$ to all the nodes in their trees based on their computed blocks $w$.
Now, parties $\cal A$ and $\cal B$ have tentative label sets $[T_{\cal A}], [T_{\cal B}] \subseteq\{0,\ldots, m\}$, respectively.
\STATE
\STATE {\bf Goal:} Change all the labels using a bijection: $[T_{\cal A}]\cup [T_{\cal B}]\to \{1,\ldots,n\}$ 
without either party having to reveal anything about their private strings.
\STATE
\STATE {\bf Notations:} $E_\mathcal{A}(x)$ denotes the ciphertext of a message $x$
encrypted by a two-level HE with $\mathcal A$'s public key.
\STATE
\STATE {\bf Sharing a dictionary:}
\STATE {\bf Step 1:}  Party ${\cal A}$ computes the bit vector ${\bf X}[1..m]$ such that
${\bf X}[\ell]=1$ iff $\ell\in [T_{\cal A}]$.
Similarly, party ${\cal B}$ computes ${\bf Y}[1..m]$ such that ${\bf Y}[\ell]=1$ iff $\ell\in [T_{\cal B}]$.
\STATE {\bf Step 2:} ${\cal A}$ sends $E_{\cal A}({\bf X})$ to ${\cal B}$ 
and ${\cal B}$ sends $E_{\cal B}({\bf Y})$ to ${\cal A}$.
\STATE {\bf Step 3:} ${\cal B}$ computes 
$
\left(E_{\cal A}({\bf X})\oplus E_{\cal A}({\bf Y})\right) \oplus
\left(E_{\cal A}({\bf X})\cdot E_{\cal A}({\bf Y})\right) = E_{\cal A}({\bf X}\cup{\bf Y})
$
and 
${\cal A}$ computes 
$
\left(E_{\cal B}({\bf X})\oplus E_{\cal B}({\bf Y})\right) \oplus
\left(E_{\cal B}({\bf X})\cdot E_{\cal B}({\bf Y})\right) = E_{\cal B}({\bf X}\cup{\bf Y})
$.
\STATE
\STATE \textbf{Relabeling $[T_\mathcal{A}]$ using $E_{\cal A}({\bf X}\cup{\bf Y})$}
($[T_\mathcal{B}]$ is relabeling in the symmetrical way)
\STATE {\bf Step 4:} ${\cal A}$ computes $E_{\cal B}(L_{\ell}) = E_{\cal B}\left(\displaystyle \sum_{i=1}^\ell ({\bf X}\cup{\bf Y}) [i]\right)$ for all $\ell\in [T_{\cal A}]$.
\STATE {\bf Step 5:} ${\cal A}$ sends all $E_{\cal B}(L_\ell + r_\ell)$ to ${\cal B}$
choosing $r_\ell$ uniformly at random from $\mathbb N$.
\STATE {\bf Step 6:} ${\cal B}$ decrypts all $L_\ell + r_\ell$ and sends them back to ${\cal A}$.
\STATE {\bf Step 7:} ${\cal A}$ recreates $L_\ell \in \{1,\ldots,n\}$ for all $\ell\in [T_{\cal A}]$
by subtracting $r_\ell$.
\STATE
\STATE
\end{algorithmic}
\end{algorithm}

Two parties ${\cal A}$ and ${\cal B}$ have strings $S_{\cal A}$ and $S_{\cal B}$, respectively.
First, they compute the corresponding ESP trees $T_{\cal A}$ and $T_{\cal B}$ offline,
using the rolling hash function to generate (tentative) consistent labels,
thereby defining a set $X\subseteq\{0,1,\ldots, m\}$ of $n$ different labels in $T_{\cal A}$ and $T_{\cal B}$
with a fixed $m$.
The algorithm's goal is to securely relabel $X$ using by a bijection: $X\to\{1,2,\ldots,n\}$, 
as described in Algorithm~\ref{algo1}, where ${\cal A}$ and ${\cal B}$ have their own public and private keys.

In our algorithm, we assume a FHE (LHE) system supporting both additive and multiplicative operations.
Since these operations are usually implemented by AND ($\cdot$) and XOR ($\oplus$) logic gates (e.g.~\cite{Brakerski2012}),
we introduce several notations using such gates as follows.
First, $E_{\cal A}(x)$ denotes the ciphertext generated by encrypting plaintext $x$ with ${\cal A}$'s public key,
and $E_{\cal A}(x,y,z)$ is an abbreviation for the vector $(E_{\cal A}(x),E_{\cal A}(y),E_{\cal A}(z))$.
Here, $E_{\cal A}(x,y,z)\cdot E_{\cal A}(a,b,c)$ denotes $(E_{\cal A}(x\cdot a),E_{\cal A}(y\cdot b),E_{\cal A}(z\cdot c))$
and 
$E_{\cal A}(x,y,z)\oplus E_{\cal A}(a,b,c)$ denotes $(E_{\cal A}(x\oplus a),E_{\cal A}(y\oplus b),E_{\cal A}(z\oplus c))$
for each bits $x,y,z,a,b,c\in\{0,1\}$.
Using these notations, we describe the proposed protocol in Algorithm~\ref{algo1}.

Next, we define our protocol's security based on a model where we assume that 
both parties are {\em semi-honest}, i.e.,
corrupted parties merely cooperate to gather information out of the protocol, 
but do not deviate from the protocol specification.
The security is defined as follows.

\begin{definition}{\rm (Semi-honest security~\cite{Goldreich2004})}
A protocol is secure against semi-honest adversaries if
%A protocol securely computes a functionality in the presence of semi-honest adversaries if
each party's observation of the protocol can be simulated using only the input they hold
and the output that they receive from the protocol.
%Corrupted parties merely cooperate to gather information from the protocol, 
%and do not deviate from its specification.
\end{definition}

Intuitively, this definition tells us that a corrupted party is unable to learn any extra
information that cannot be derived from the input and output explicitly
(For details, see~\cite{Goldreich2004}).
Under this assumption, 
since the algorithm is symmetric with respect to $\mathcal A$ and $\mathcal B$,
the following theorem proves our algorithm's security against semi-honest adversaries.

\begin{theorem}\label{th1}\rm
  Let $[T_{\cal A}]$ be the set of labels appearing in $T_{\cal A}$.
  The only knowledge that a semi-honest $\mathcal A$ can gain by executing Algorithm~\ref{algo1}
  is the distribution of the labels $\{L_\ell \mid \ell \in [T_\mathcal{A}]\}$ over $[1, \dots, n]$.
\end{theorem}
\begin{proof}
  First, the preprocessing phase gives $\cal A$ no new information,   since it is conducted offline.
  Second, the dictionary sharing phase does not provide any new knowledge either,
  since all the information that $\mathcal A$ receives from $\mathcal B$ is
  encrypted using $\mathcal B$'s public key.
  Third, when $\cal A$ is relabeling $[T_\mathcal{A}]$,
  they only receive $L_\ell$ for $\ell \in [T_\mathcal{A}]$.
  Finally, when $\cal B$ is relabeling $[T_\mathcal{B}]$,
  $\mathcal A$ knows $L_\ell + r_\ell$ for $\ell \in [T_\mathcal{B}]$,
  but the $r_\ell$ are secret random numbers that $\mathcal B$ has generated uniformly at random,
  and $\mathcal A$ cannot know their values.
  Hence, the $L_\ell + r_\ell$ are distributed uniformly at random from $\mathcal A$'s perspective.
\hspace{\fill}$\Box$
\end{proof}

Although $\mathcal A$ can guess $n$ as being either $\max\{L_\ell \mid \ell \in [T_\mathcal{A}]\}$ or 
a value just above this, and can obtain knowledge about $\cal B$'s labels by investigating 
$\{1, \dots, n\} \setminus \{L_\ell \mid \ell \in [T_\mathcal{A}]\}$,
since we have assumed that the hash function is (probabilistically) one-way, this does not give $\mathcal A$ 
any knowledge about $\cal B$'s text.

\begin{theorem}\label{th2}\rm
Algorithm~\ref{algo1} assigns consistent labels using the injection: $[T_{\cal A}]\cup [T_{\cal B}]\to \{1,2,\ldots,n\}$
without revealing the parties' private information.
Its round and communication complexities are $O(1)$ and $O(\alpha(n\lg n+m+rn))$, respectively,
where $n=|[T_{\cal A}]\cup[T_{\cal B}]|$, $m$ is the modulus of the rolling hash used for preprocessing,
$r=\max\{r_1,\ldots,r_n\}$ is the security parameter, and
$\alpha$ is the cost of executing a single encryption, decryption, or homomorphic operation.
\end{theorem}
\begin{proof}
A two-level HE scheme allows sufficient number of additions and a single multiplication of encrypted integers.
Thus, we can securely represent the set $[T_{\cal A}]\cup [T_{\cal B}]$ by
$E_{\cal A}({\bf X}\cup{\bf Y})=\left(E_{\cal A}({\bf X})\oplus E_{\cal A}({\bf Y})\right) \oplus
\left(E_{\cal A}({\bf X})\cdot E_{\cal A}({\bf Y})\right)$,
and assign consistent labels $L_\ell = \mathrm{\bf rank}_1(\ell, {\bf X}\cup{\bf Y})$
for all $\ell\in [T_{\cal A}]\cup [T_{\cal B}]$.
Thus, the parties can securely obtain consistent labels, with the security level depending on the encryption strength.
Regarding the communication complexity, the plaintexts sent have size of $m$ bits (Step2) and $n\lg n$ bits (Step5),
from which we can immediately derive the complexity using the other parameters.
The round complexity is evident.
\hspace{\fill}$\Box$
\end{proof}

%% file: sec4.tex
\section{Experimental Results}

\begin{table*}[t]
\begin{center}
\caption{
Execution time (seconds) comparison for Phase 1, showing the preprocessing and relabeling time per label
for the number $n$ of characteristic substrings to be relabeled.
Here, ``Preprocessing'' denotes the time required to construct the shared dictionary
and ``Relabeling (per label)'' denotes the time needed to change a single label using the dictionary.
} 
%\vspace{1mm}
\label{tab2}
\begin{tabular}{lrrr}
\hline
& \qquad\qquad $n$ \qquad & \qquad sEDM~\cite{Nakagawa2018} \qquad\qquad & Ours \qquad \\ \hline
  & 100 & \qquad 9.772 \qquad\qquad  & \qquad 3.147 \qquad  \\
 Preprocessing \qquad & 1000\qquad &\qquad 76.996 \qquad\qquad  & \qquad 31.150 \qquad \\ 
& 10000\qquad & \qquad 725.463 \qquad\qquad & \qquad 304.314 \qquad \\ 
& 100000\qquad & \qquad 7264.354 \qquad\qquad & \qquad 3030.031 \qquad \\ \hline
 & 100 & \qquad 13.977 \qquad\qquad & \qquad 0.010 \qquad \\
 Relabeling (per label) \qquad & 1000\qquad &  \qquad 160.995 \qquad\qquad & \qquad 0.047 \qquad \\ 
& 10000\qquad & \qquad 1066.259 \qquad\qquad & \qquad 0.319 \qquad \\ 
& 100000\qquad & \qquad NA ($>10000$) \qquad\qquad & \qquad 2.124 \qquad \\ \hline
\end{tabular}
\end{center}
%\vspace{.5mm}
\end{table*}

\begin{table*}[t]
\begin{center}
\caption{
Execution time (seconds) of approximated EDM computation for Escherichia coli (100MB).
Here, $n$ is the number of characteristic substrings used for EDM and
the same rolling hash in Table~\ref{tab2} is used for each $n$.
$L_1$-distance is computed by the sEDM~\cite{Nakagawa2018}.
} 
%\vspace{1mm}
\label{tab3}
\begin{tabular}{lrr}
\hline
Detail of EDM computation & \qquad\qquad $n$ \qquad & time \qquad \\ \hline
  & 100 & \qquad 4.055 \qquad  \\
 Relabeling by our algorithm (Phase 1) \qquad & 1000 & \qquad 63.597 \qquad \\ 
& 10000\qquad & \qquad 2506.431 \qquad \\  \hline
%& 100000\qquad & \qquad ? \qquad\qquad & \qquad ? \qquad \\ \hline
 & 100 & \qquad 4.097 \qquad \\
 $L_1$-distance computation by sEDM~\cite{Nakagawa2018} (Phase 2) \qquad & 1000\qquad & \qquad 4.135 \qquad \\ 
& 10000\qquad & \qquad 4.689 \qquad \\  \hline
%& 100000\qquad & \qquad - \qquad\qquad & \qquad ? \qquad \\ \hline
\end{tabular}
\end{center}
%\vspace{.5mm}
\end{table*}

Finally, we compared the practical performance of our algorithm with that of sEDM~\cite{Nakagawa2018}. 
Both algorithms were implemented in C++ based on the two-level HE~\cite{Attrapadung2018} and 
library available from GitHub\footnote{
{\tt https://github.com/herumi/mcl}}, and 
compiled using Apple LLVM version 8.0.0 (clang-800.0.42.1)
under MacOS Mojave 10.14.5.
The algorithm's performance was evaluated on a system with a 2.7 GHz Intel Core i5 CPU
and 8 GB 1867 MHz DDR3 RAM.

Table~\ref{tab2} shows the results for generating $n\in\{100,1000, 10000,100000\}$ different labels
shared by the parties.
%where $n=|[T_{\cal A}]\cup[T_{\cal B}]|$.
The key length of encryption is fixed to 256 bits.
Our algorithm uses the rolling hash modulo $p\in\{1031, 10313, 103123, 1031347\}$ for each $n$, respectively.
This shows the running times for each $n$, where ``Preprocessing'' gives the time $t_1$ required to construct the shared dictionary
and ``Relabeling (per label)'' gives the response time $t_2$ needed to change a single label.
Thus, the total time for each algorithm is $t_1+nt_2$ for each $n$.
Note that the total time is mainly occupied by the relabeling time for both algorithms.
Therefore, these results confirm that our algorithm's computation was significantly lower than that of sEDM in all cases.

Table~\ref{tab3} shows the total time of approximate EDM computation 
for real DNA sequence 
available from Pizza\&Chili Corpus\footnote{{\tt http://pizzachili.dcc.uchile.cl}}.
From this corpus, we use Escherichia coli, known as {\em highly repetitive string}
where 110MB original string is compressed to 5MB by 7-zip\footnote{{\tt https://www.7-zip.org/ }}.
This means that the number of characteristic substrings is relatively smaller,
so the restriction of the examined $n$ up to 10000 is reasonable.
However, in reality, we cannot execute the sEDM~\cite{Nakagawa2018} even for these repetitive strings
due to the cost of relabeling shown in Table~\ref{tab2}.
By the results in Table~\ref{tab2} and~\ref{tab3},
we confirm the efficiency of our algorithm for large-scale data.

%% file: sec5.tex
\section{Conclusion}
In this paper, we have presented an improvement to a 
previously proposed HE-based secure two-party protocol for computing approximate EDM.
The problem we tackled is reduced to jointly assigning minimum consistent labels
from $X\cup Y\subseteq\{1,2,\ldots,m\}$ to $\{1,2,\ldots,n\}$.
The fact that recent two-level HE systems allow sufficient number of additions and a single multiplication
over ciphertexts enabled us to significantly improve the execution time.
From a cryptographic point of view, $m$ should be sufficiently large (i.e., we assume that $X\cup Y$ is sparse)
so that it is difficult for Alice to learn about Bob's labels from the distribution of ones in his label vector.
In contrast, $m$ should be smaller for saving the communication cost.
We plan to investigate this problem further in future work.

To the best of our knowledge, existing two-party protocols for consistent labeling need both additive and 
multiplicative homomorphic operations over the ciphertexts.
Since an HE system that only involves additive operations is computationally less taxing,
whether or not we can solve the relabeling problem by only exploiting  
additive homomorphism is an important practical question.